\documentclass[11pt,letterpaper]{article}
\usepackage{color}
\usepackage{hyperref}
\usepackage{amssymb}
\usepackage{amsmath}
\usepackage{amsthm}
\usepackage{amstext}
\usepackage{fullpage}
\usepackage{epsfig}
\usepackage{mytex}

\usepackage{color} % for debug

\newtheorem{theorem}{Theorem}
\newtheorem*{theorem*}{Theorem}
\newtheorem{fact}[theorem]{Fact}
\newtheorem{proposition}[theorem]{Proposition}
\newtheorem{claim}[theorem]{Claim}
\newtheorem{corollary}[theorem]{Corollary}
\newtheorem*{corollary*}{Corollary}
\newtheorem{lemma}[theorem]{Lemma}
\theoremstyle{definition}
\newtheorem{definition}[theorem]{Definition}
\newtheorem*{remark}{Remark}

\newreptheorem{theorem}{Theorem}
\newreptheorem{corollary}{Corollary}
\newreptheorem{lemma}{Lemma}

\begin{document}

\newcommand{\ignore}[1]{}

\title{A Duality Between Depth-Three Formulas and \\ Approximation by Depth-Two}
\date{}
\author{
Shuichi Hirahara%
\footnote{Department of Computer Science, The University of Tokyo, Tokyo, Japan, \texttt{\href{mailto:hirahara@is.s.u-tokyo.ac.jp}{hirahara@is.s.u-tokyo.ac.jp}}}
}
%\\
%The University of Tokyo}

\maketitle

\newcommand{\cor}[1][]{D_{#1}}
\newcommand{\CNF}[1]{\mathrm{CNF}_{#1}}
\newcommand{\PAR}{\textsc{Parity}}
\newcommand{\MAJ}{\textsc{Maj}}
\newcommand{\lessthan}[2]{{#1}^{-1}(1) \subset {#2}^{-1}(1)}

\newcommand{\headline}[1]{\subsection{#1}}
\newcommand{\smallheadline}[1]{\vspace{0.0em}\textbf{#1.}}

\begin{abstract}
  We establish an explicit link 
  between
  depth-3 formulas and
  one-sided-error approximation by depth-2 formulas,
  which were previously studied independently.
  Specifically,
  we show that 
  the minimum size of
  depth-3 formulas
  is (up to a factor of $n$) equal to
  the inverse of 
  the maximum, over all depth-2 formulas, of
  one-sided-error correlation bound divided by the size of the depth-2 formula,
  on a certain hard distribution.
  We apply this duality to obtain several consequences:
  \begin{enumerate}
    \item
      Any function $f$ can be approximated by a CNF formula of size $O(\epsilon 2^n / n)$
      with one-sided error and advantage $\epsilon$ for some $\epsilon$,
      which is tight up to a constant factor.
    \item
      There exists a monotone function $f$ such that
      $f$ can be approximated by some polynomial-size CNF formula,
      whereas any \emph{monotone} CNF formula approximating $f$ requires exponential size.
    \item
      Any depth-3 formula computing the parity function requires 
      $\Omega(2^{2 \sqrt{n}})$ gates,
      which is tight up to a factor of $\sqrt n$.
      This establishes a quadratic separation
      between depth-3 \emph{circuit} size and depth-3 \emph{formula} size.
    \item
      We give a characterization 
      of the depth-3 monotone circuit complexity of the majority function,
      in terms of 
      a natural extremal problem on hypergraphs.
      In particular, we show that
      a known extension of Tur\'an's theorem
      gives 
      a tight (up to a polynomial factor) circuit size for
      computing the majority function by
      a monotone depth-3 circuit with bottom fan-in 2.
    \item
      $\AC^0[p]$ has exponentially small one-sided correlation
      with the parity function for odd prime $p$.
  \end{enumerate}
\end{abstract}

%\thispagestyle{empty}
%\newpage
%\setcounter{page}{1}

\section{Introduction}

The main theme of this paper is a new connection between
approximation by depth-2 formulas
and
exact computation of depth-3 formulas.
We first review these two (previously independent) lines of research.

\headline{Approximation by Depth-2 Formulas}
A depth-2 formula is a CNF formula or a DNF formula.
By De Morgan's law, we can often assume without loss of generality that 
a depth-2 formula is a CNF formula; hence, in this paper, we will focus on CNF formulas.
The size $|\phi|$ of a CNF formula $\phi$ is the number of clauses in the formula.

A CNF formula is one of the simplest computational models in complexity theory,
and its computational power is well understood;
for example, the parity function requires CNF formulas of size $2^{n-1}$, which is exactly tight;
any function can be computed by a CNF formula of size $2^{n-1}$, and hence the parity function
is the hardest problem against a CNF formula.

Nonetheless, the situation becomes quite different when,
instead of exact computation, 
\emph{approximation} by CNF formulas is concerned.
Namely, we allow a CNF formula to err on some fraction of inputs.
We say that a CNF formula $\phi$ $\epsilon$-approximates a function $f$
if $\Pr[f(x) \neq \phi(x)] < \epsilon$.
Initiated by O'Donnell and Wimmer \cite{OW07},
a line of research \cite{OW07,BT15,BHST14}
has highlighted 
surprising differences between
exact computation of depth-2 formulas
and 
approximation by depth-2 formulas.
Two of them are reviewed below.

\smallheadline{Quine's Theorem for Approximation}
Blais, H{\aa}stad, Servedio, and Tan \cite{BHST14} studied
whether an analog of Quine's theorem holds.
Quine's theorem states that the minimum CNF formula computing a monotone function is 
achieved by monotone CNF formulas (\ie CNF formulas without any negated literal).
\begin{theorem*}
  [Quine \cite{Qui54}]
  A smallest CNF formula computing a monotone function \emph{exactly}
  is monotone as well.
\end{theorem*}
Given this fact, it is tempting to guess that 
a smallest CNF formula \emph{approximating} a monotone function
is monotone.
However, in \cite{BHST14}, it was shown that there is
at least a quadratic gap between
the size of
the smallest CNF formula
and
the smallest monotone CNF formula
approximating a monotone function $f$.
\begin{theorem*}
  [\cite{BHST14}]
  There are a parameter $\epsilon(n)$ and a monotone function on $n$ variables that can be $\epsilon(n)$-approximated by some \emph{nonmonotone} CNF formula of size $O(n)$,
  but cannot be $\epsilon(n)$-approximated by any \emph{monotone} CNF formula of size less than $n^2$.
\end{theorem*}
In this paper, we will significantly improve this bound and exhibit an \emph{exponential} separation.

\smallheadline{Universal Bounds}
Blais and Tan \cite{BT15} studied the universal bound on the size of CNF formulas approximating a function:
they
showed that any function $f$ can be $\epsilon$-approximated by
any CNF formula of size $O_\epsilon(2^{n} / \log n)$
for any constant $\epsilon$ (where a constant $O_\epsilon$ depends on $\epsilon)$.
They also gave a lower bound of $\Omega_\epsilon(2^{n} / n)$ for a CNF formula $\epsilon$-approximating
a random function, and thus leaving a gap.
Namely, the maximum complexity of CNF formulas approximating functions was not well understood.

\headline{Depth-3 Circuits}
As with approximation by depth-2, 
another simple computational model whose power is still ``mystery'' in complexity theory
is \emph{depth-3 circuits} (or formulas).
A depth-3 ($\AC^0$) circuit here is a directed acyclic graph
consisting of alternating 3 layers of AND and OR gates
and an input layer whose gates are labelled by literals (\ie an input $x_i$ or its negation $\neg x_i$).
A depth-3 \emph{formula} is a depth-3 circuit whose gates have fan-out $1$
(\ie a computational model in which the intermediate computation cannot be reused).
For simplicity, we assume that the top gate of depth-3 circuits and formulas is an OR gate (\ie $\OR \circ \AND \circ \OR$ circuits).

A counting argument \cite{JS42} shows that 
most functions $f$ requires $\Omega(2^n / \log n)$ literals for a formula to compute $f$.
And this is tight for depth-3 formulas, as a classical theorem by Lupanov \cite{Lup65}
shows that any function can be computed by a depth-3 formula with $O(2^n / \log n)$ literals.
Therefore, the complexity of depth-3 formulas computing \emph{random functions}
is fairly well understood.

In sharp contrast, it is wide open to prove a lower bound $2^{\omega(\sqrt{n})}$
of a depth-3 circuit (or formula)
computing some \emph{explicit}
function (\eg a function in $\P$, or even $\E^\NP$).
In 1980s,
there has been significant progress in the understanding
of the computational power of unbounded fan-in constant depth circuits
(\eg \cite{Ajt83, FSS84, Yao85, Has89, Raz87, Smo87}).
These results give
a depth-$d$ circuit lower bound of the form $2^{\Omega(n^{1/(d-1)})}$ for explicit functions,
such as the parity and majority functions.
After 30 years, this remains asymptotically the current best circuit lower bounds against constant depth circuits.

There are several reasons why depth-3 circuits are interesting.
One notable reason
was given by Valiant \cite{Val77},
who showed that
any linear-size log-depth circuit
can be transformed into a
depth-3 circuit of size $2^{O(n / \log \log n)}$.
Hence,
a strongly exponential depth-3 circuit lower bound
will yield a super-linear log-depth circuit lower bound simultaneously,
the former of which
appears to be easier to deal with.
Considerable efforts (\eg \cite{HJP95, PPZ99, PSZ00, PPSZ05, IPZ01, GW13, GT16}) have been thus made
to obtain strongly exponential depth-3 circuit lower bounds.

Another reason is that depth-3 circuits are closely related to CNF-SAT algorithms.
For example,
Paturi, Pudl{\'{a}}k and Zane \cite{PPSZ05}
proved that the minimum size of a depth-3 circuit computing the parity function is $\Theta(n^{1/4} 2^{\sqrt n})$,
and simultaneously developed a simple and fast CNF-SAT algorithm;
Paturi, Pudl{\'{a}}k, Saks and Zane \cite{PPSZ05}
improved the CNF-SAT algorithm
and simultaneously gave
a depth-3 circuit lower bound $2^{1.282 \ldots \sqrt n}$ for some explicit function,
which is the current best depth-3 circuit lower bound.
(A relationship between a fast CNF-SAT algorithm and a circuit lower bound
was made formal by Williams \cite{Wil13}.)

\headline{Duality Theorem}

The main theme of this paper is to weave these two lines of research together;
thereby we advance the two lines simultaneously.
To this end, we will exploit a general result below, which
establishes an equivalence between exact computation of $\OR \circ \mathcal C$
and one-sided-error approximation by $\mathcal C$ for any circuit class $\mathcal C$.
(Here, $\OR \circ \mathcal C$ denotes 
the class of circuits that consist of a top OR gate fed by disjoint $\mathcal C$ circuits.
In most of our applications, we take $\mathcal C = \ESet{ \text{CNF formulas} }$
and hence $\OR \circ \mathcal C$ is the class of depth-3 formulas.)

First, we need to introduce one new notion.
For any circuit class $\mathcal C$ and a Boolean function $f$ and a distribution $\mu$ on $f^{-1}(1)$,
define \emph{maximum (one-sided-error) correlation%
  \footnote{In a usual context of circuit lower bounds, a correlation between functions $f$ and $g$ is defined 
as $2\Pr[f(x) = g(x)] - 1$.
While the definition here is slightly different, we borrow the terminology.}
per size} of $f$ as 
\[
  \cor[\mu]^{\mathcal C}(f) := \max_{ \substack{ \phi \in \mathcal C \\ \phi \le f }} \Pr_{x \sim \mu} [\phi(x) = f(x)] / |\phi| \supplement{ = \max_{ \substack{ \phi \in \mathcal C \\ \phi \le f } }  \Pr_{x \sim \mu} [\phi(x) = 1] / |\phi| },
\]
and define $\cor^{\mathcal C}(f) := \min_\mu \cor[\mu]^{\mathcal C}(f)$.
Intuitively, (one-sided-error) correlation per size measures 
a trade-off between the size of a circuit $\phi \in \mathcal C$
and how well $f$ can be approximated by a circuit $\phi \in \mathcal C$
with one-sided error in the sense that $\lessthan \phi f$.
Our general result relates the complexity of exact computation of $\OR \circ \mathcal C$
and the complexity of one-sided-error approximation by $\mathcal C$:
\begin{theorem}
  [Duality (informal)]
  For any circuit class $\mathcal C$ and a Boolean function $f \colon \binset^n \to \binset$,
  let $\mathrm{L}^{\mathcal C}(f)$ denote the minimum size of an $\OR \circ \mathcal C$ circuit computing $f$.
  Then, 
  \[
    \mathrm{L}^{\mathcal C}(f)
    \approx
    (\cor^{\mathcal C}(f))^{-1},
  \]
  up to a factor of $n$.
\end{theorem}
Here, the size of an $\OR \circ \mathcal C$ circuit $C_1 \vee \dots \vee C_m$
is defined as the sum $\sum_{i=1}^m |C_i|$.

In fact, its proof is quite simple:
Regard a computation by an $\OR \circ \mathcal C$ circuit as a set cover instance,
consider the linear programming relaxation, and then
take its \emph{dual}, whose optimal value corresponds to $(\cor^{\mathcal C}(f))^{-1}$.

\smallheadline{Related Work on Duality}
It should be noted that 
a similar result was known before
in the context of threshold circuits.
Indeed, the discriminator lemma by Hajnal, Maass, Pudl{\'{a}}k, Szegedy and Tur{\'{a}}n \cite{HMPST93}
states that
a circuit with top threshold gate that computes a function $f$
must have a subcircuit which has a high correlation with $f$.
A converse direction of the discriminator lemma was
proved by Goldmann, H{\aa}stad, and Razborov \cite{GHR92}
in the context of threshold circuits,
and by Freund \cite{Fre95} in the context of boosting.
Our duality can be regarded as a version of the discriminator lemma and its converse specialized 
to the case of circuits with top AND or OR gate (instead of threshold gate),
with a \emph{tighter quantitative trade-off} between correlation and size.
To the best of our knowledge, the notion of correlation per size
was neither defined nor recognized as important quantity before;
and the direction $\mathrm{L}^{\mathcal C}(f) \lesssim (\cor^{\mathcal C}(f))^{-1}$ was not known before.

In what follows, we explain how our duality theorem advances the lines of research on
depth-3 formulas and
approximation by depth-2 formulas.

\headline{Universal Bounds on Depth-2 Approximator}
We first apply the duality in order to obtain a tight universal bound on a depth-2 approximator.
As mentioned before, Blais and Tan \cite{BT15} left a gap between an upper bound $O(2^n / \log n)$ of 
a CNF formula approximating any function
and a lower bound $\Omega(2^n / n)$ for a random function.

In contrast, the depth-3 formula complexity of random functions is well understood \cite{JS42,Lup65}, at least
when the size of depth-3 formulas is defined as the number of \emph{literals} (and the bound is $\Theta(2^n / \log n)$).
For our purpose, we need to redefine the size of a depth-3 formula
as the number of AND gates at the middle layer, so that it is consistent with the fact that
the size of a CNF formula is measured as the number of clauses.
In fact, it turns out that Lupanov's construction \cite{Lup65} can be adapted to this size measure,
and we obtain an improved upper bound $O(2^n / n)$ of depth-3 formulas computing any function.
Our duality theorem transfers this result into approximation by depth-2, and we obtain the following tight bound:
\begin{theorem}
  [informal]
  \label{thm:universalapproximator_informal}
  \begin{itemize}
    \item
      For any function $f$,
      there exists some CNF formula $\phi$
      of size $O(\epsilon \cdot 2^n / n)$
      approximating $f$ with one-sided error and advantage $\epsilon$ for some $\epsilon$.
    \item
      There exists a function $f$
      such that
      any CNF formula approximating $f$ with one-sided error and advantage $\epsilon$
      must have size $\Omega(\epsilon \cdot 2^n / n)$.
  \end{itemize}
\end{theorem}
Here, the constants hidden in $O$ and $\Omega$ are universal (and, in particular, do not depend on $\epsilon$ nor $f$).

\headline{Approximating Monotone Functions}
Next, we apply the duality theorem in order to obtain
an exponential separation 
between
nonmonotone CNF formulas
and monotone CNF formulas
computing a monotone function,
which significantly improves \cite{BHST14}:
\begin{theorem}
  [informal]
  \label{thm:monotoneapprox_informal}
  There exists a monotone function $f$ such that 
  \begin{enumerate}
    \item
      there is some polynomial-size CNF formula that approximates $f$ with one-sided error, whereas
    \item
      any \emph{monotone} CNF formulas approximating $f$ with one-sided error 
      requires exponential size.
  \end{enumerate}
\end{theorem}
That is, Quine's theorem fails badly for approximation by depth-2 formulas.

Our proof is in fact an immediate consequence of a recent result by
Chen, Oliveira and Servedio \cite{COS15}:
they showed that there is a monotone function $f$ such that
$f$ can be computed by some depth-3 formula of polynomial size,
whereas any \emph{monotone} depth-3 formula computing $f$ requires exponential size.
Our duality theorem transfers their result to Theorem \ref{thm:monotoneapprox_informal}.

\headline{Depth-3 Formula Lower Bound for Parity}
In the second half of this paper, we regard the duality theorem as 
a general approach for understanding the depth-3 formula complexity.
That is, we aim at proving an upper bound on the correlation per size $\cor(f)$ of a CNF formula,
and use it to obtain a depth-3 formula lower bound.
%We prove an almost tight lower bound of a depth-3 formula computing the parity function.
We prove that 
any depth-3 formula computing 
the parity function
requires $\Omega(2^{2 \sqrt n})$ gates,%
\footnote{Independently of our work, Rahul Santhanam and Srikanth Srinivasan (personal communication) obtained the same lower bound.}
which is tight up to a factor of $\sqrt n$.

Our proof follows from an almost tight one-sided-error correlation bound,
and our correlation bound improves another result of Blais and Tan \cite{BT15}:
They studied
the minimum size of CNF formulas that compute $\PAR_n$
all but an $\epsilon$ fraction of inputs.
They showed an upper bound of $2^{(1 - 2^{-\ceil{\log 1/2\epsilon}} n)} \approx 2^{(1 - 2\epsilon) n}$
(moreover, with one-sided error in our sense)
and an lower bound of $(\frac{1}{2} - \epsilon) 2^{\frac{1 - 2\epsilon}{1 + 2\epsilon} n}$.
In this paper, we focus on approximation with one-sided error and high error regimes
(\eg $\epsilon = \frac{1}{2} - 2^{-\sqrt n}$).
By using the satisfiability coding lemma \cite{PPZ99} and width reduction techniques \cite{Sch05, CIP06},
we obtain an lower bound of $2^{n / (k + 1) - 3}$ for $k := \log \frac{2}{1 - 2\epsilon}$
and an upper bound of $k' 2^{\ceil{n / k'} - 1}$ for $k' := \floor{k}$,
which significantly improves the results of \cite{BT15} when 
the error fraction $\epsilon$ is close to $\frac{1}{2}$.

In terms of correlation bounds,
our result shows that
any CNF formula $\phi$ that computes $\PAR_n$ and does not err on inputs in $\PAR_n^{-1}(0)$
has the correlation with $\PAR_n$ at most $2^{- n / (\log |\phi| + 3) + 2}$.
This improves the previous bounds 
$2^{- n / O(\log s)^{d-1}}$
on the correlation between the parity function and depth-$d$ circuits of size $s$
in the case of $d = 2$
(Beame, Impagliazzo and Srinivasan \cite{BIS12}, H{\aa}stad \cite{Has14} and Impagliazzo, Matthews, and Paturi \cite{IMP12}).

Given the almost optimal one-sided-error correlation bound,
our duality theorem immediately implies a
depth-3 formula lower bound as follows:
The one-sided-error correlation per size of a CNF formula $\phi$ with $\PAR_n$
is at most $2^{- n / (\log |\phi| + 3) + 2} / |\phi| = 2^{5 - n / (\log |\phi| + 3) - (\log |\phi| + 3)}$,
which is bounded above by $2^{5 - 2\sqrt{n}}$
from the inequality of the arithmetic and geometric means.
Hence, any depth-3 formula computing the parity function requires $2^{2\sqrt{n} - 5}$ gates.

\headline{Depth-3 Circuits vs.\ Formulas}
Our formula lower bound on the parity function is of interest
from yet another literature:
Questions whether circuits are strictly more powerful than formulas 
are one of central questions in complexity theory.
It is straightforward to see that any unbounded fan-in depth-$d$ circuit of size $s$
has an equivalent representation as a depth-$d$ formula of $s^{d-1}$,
by simply replicating overlapping subcircuits.
(Thus, depth-3 formula size and depth-3 circuit size are at most quadratically different.)
The converse direction, namely, whether this \naive simulation is optimal,
is closely related to the $\NC^1$ vs.\ $\AC^1$ problem.
Indeed, if one could show
that there exists a language that can be computed by log-depth
unbounded fan-in polynomial-size circuits (\ie $\AC^1$ circuits)
but requires a formula%
\footnote{Note that, by Spira's theorem \cite{Spi71}, $\NC^1$ can be characterized as languages 
computable by polynomial-size formulas. } of size $n^{\Omega(\log n)}$, then $\NC^1 \neq \AC^1$.
Although
we thus cannot hope that the circuit vs.\ formula question can be solved
for $d = O(\log n)$ anytime soon
(due to the natural proof barrier \cite{RR97}),
the question was solved affirmatively in some restricted cases.

The monotone circuit vs.\ monotone formula question
was solved by Karchmer and Wigderson \cite{KW90}.
They showed that monotone formulas computing st-connectivity require $n^{\Omega(\log n)}$ gates,
whereas there is a polynomial-size monotone circuit computing st-connectivity.
Their communication complexity theoretic approach (Karchmer-Wigderson games)
has been quite successful in monotone settings.
However, little was known in nonmonotone settings
until 
recent results by Rossman \cite{Ros14, Ros15}
(see \cite{Ros14} for further background).
He showed that the simulation of depth-$d$ circuits of size $s$ by 
depth-$d$ formulas requires $s^{\Omega(d)} \supplement{> s}$ gates
for a sufficiently large $d$ (say, $d \ge 108$) up to $d = o(\frac{\log n}{\log \log n})$.
More specifically, he \cite{Ros15} showed that 
any depth-$d$ formula computing $\PAR_n$ requires $2^{\Omega(d (n^{1 / (d-1)} - 1))}$ gates,
whereas it is known that
$\PAR_n$ can be computed by a depth-$d$ circuit of size $n 2^{n^{1 / (d-1)}}$.

Motivated by Rossman's results,
we may ask whether depth-$d$ circuits are more powerful than depth-$d$ formulas for a small constant $d$.
And we may ask whether depth-$d$ circuits of size $s$
\emph{cannot} be simulated by depth-$d$ formulas of size even slightly better than the \naive simulation, say, $s^{d - 1.01}$.
We answer these questions affirmatively for $d = 3$:
Our formula lower bound immediately implies that
simulating depth-3 circuits by depth-3 formulas
requires a \emph{quadratic} slowdown,
thereby separating depth-3 circuit and formula size almost optimally.
\begin{corollary*}
  Let $s ( = \Theta(n^{1/4} 2^{\sqrt n}) )$ be the minimum
  depth-3 circuit size for computing $\PAR_n$.
  Any depth-3 formula computing $\PAR_n$ requires $\Omega(s^2 / \log s)$ gates.
\end{corollary*}
\noindent
We note that this does not seem to follow from Rossman's techniques
as he used the switching lemma \cite{Has89}, which may lose some constant
factor in the exponent.

\headline{Computing Majority by Depth-3 Circuits}
We also make a step towards better understanding of 
depth-3 formulas computing the majority function.
We characterize the monotone complexity of the majority function
in terms of a natural extremal problem on hypergraphs:
For a hypergraph $\mathcal F$,
let
$|\mathcal F|$ denote the number of edges in $\mathcal F$,
$\tau(\mathcal F)$ denote the minimum size of hitting sets of $\mathcal F$
(\ie sets that intersect with every edge in $\mathcal F$),
and $t(\mathcal F)$ denote the number of minimum hitting sets of $\mathcal F$.
Let $T(n, \tau) := \max_{\mathcal F \colon \tau(\mathcal F) = \tau} t(\mathcal F) / |\mathcal F|$,
where the maximum is taken over all hypergraphs of $n$ vertices.
In Section \ref{sec:hard}, we show that $2^n / T(n, n/2)$ is asymptotically
equal to the minimum size of monotone depth-3 formulas computing the majority function.
The main idea of the proof is that one of the hardest distributions against monotone CNF formulas
can be exactly determined.

It is known that any depth-3 formula (or circuit) computing the majority function
requires $2^{\Omega(\sqrt n)}$ gates (\eg \cite{Has89, HJP95}),
whereas the current best upper bound on the size of a depth-3 formula computing 
the majority function is $2^{O(\sqrt{n \log n})}$ (which can be constructed by partitioning variables into blocks of equal size; see \cite{KPPY84}),
and the depth-3 formula is monotone.
Thus, it follows that $2^{n - O(\sqrt{n \log n})} \le T(n, n/2) \le 2^{n - O(\sqrt{n})}$.
To the best of our knowledge, there has been essentially no improvement
on the minimum depth-3 circuit computing the majority function
over 20 years since H{\aa}stad \etal \cite{HJP95} posed the question explicitly,
even when restricted to the case of \emph{monotone} circuits.
We propose an open problem of determining the asymptotic behavior of $T(n, n/2)$
as the first step towards determining the minimum depth-3 nonmonotone
circuits computing the majority function.

We note that, in the case of graphs (instead of hypergraphs),
$T(n, \tau)$ and its extremal structure are well understood in the context of extensions of Tur\'an's theorem \cite{Tur41}.
We show that 
a known extension of Tur\'an's theorem
implies the optimal circuit lower bound for
computing the majority function by monotone depth-3 circuits with bottom fan-in 2.

\headline{Worst-case to Average-case Connection}
Our duality theorem can be viewed
as a \emph{generic} equivalence
between
worst-case complexity and one-sided-error average-case complexity:
It applies to any computational model 
that is capable of simulating an unbounded fan-in OR or AND gate.

As a concrete application, we apply it to the case of $\AC^0[p]$.
Here, $\AC^0[p]$ denotes the class of constant-depth circuits consisting of
NOT gates and unbounded fan-in AND, OR, and $\MOD_p$ gates.
Razborov \cite{Raz87} and Smolensky \cite{Smo87} established
celebrated exponential lower bounds of $\AC^0[p]$ circuits computing $\PAR_n$
for odd prime $p$.
Their techniques also give an average-case lower bound such that
any $\AC^0[p]$ circuit of size at most $2^{-n^{o(1)}}$ 
has correlation with $\PAR_n$ at most $n^{-1/2 + o(1)}$
(Smolensky \cite{Smo93}; see also \cite{Fil10}).
This correlation bound remains the strongest known,
and it is a long-standing open problem whether the correlation bound can be improved.
See \cite{FSUV13} for more detailed backgrounds.

While we were not able to obtain (two-sided-error) correlation bounds,
we do obtain a \emph{one-sided-error} correlation bound.
Our duality theorem
transfers 
the worst-case lower bound of Razborov and Smolensky \cite{Raz87,Smo87}
against $\OR \circ \AC^0[p]$ circuits
to the following one-sided-error correlation bound:
\begin{theorem}
  \label{thm:worst_to_average}
  Let $C$ be any $\AC^0[p]$ circuit of depth $d$ for some odd prime $p$
  such that $C^{-1}(1) \subset \PAR_n^{-1}(1)$.
  Then we have $\Pr_{x \sim \binset^n} [ C(x) = \PAR_n(x) ] \le \frac{1}{2} + |C| \cdot 2^{-n^{\Omega(1/d)}}$,
  where $|C|$ denotes the number of gates of $C$.
\end{theorem}
\noindent
In particular, any $\AC^0[p]$ circuit of depth $d$ and size at most $2^{n^{O(1/d)}}$
that does not err on $\PAR_n^{-1}(0)$
has correlation with $\PAR_n$ at most $2^{-n^{\Omega(1/d)}}$.

\headline{Preliminaries and Notations}
Throughout this paper,
we assume that the top gate of a depth-3 circuit or formula is an OR gate.
Note that this assumption does not lose any generality for computing the parity or majority function, due to De Morgan's laws.
A \emph{depth-3 formula} is thus
an OR of CNF formulas,
and its size
is the sum of the size of CNF formulas in the depth-3 formula.
A \emph{CNF formula} is an AND of clauses,
and its size is the number of clauses in the CNF formula.
A clause is an OR of literals.
(Alternatively,
the size of depth-3 formulas is
defined to be 
the number of OR gates at the bottom layer when represented as a rooted tree.
Note that our measure is the same with the number of gates up to a factor of $2$
and is the same with the number of literals up to a factor of $2n$.)
For a function $f \colon \binset^n \to \binset$,
we write $\formulathree(f)$ for the minimum size of depth-3 formulas
computing $f$.
We often identify a formula with the function computed by the formula.

$\PAR_n \colon \binset^n \to \binset$ is the function
such that $\PAR_n(x) = 1$ if and only if the number of ones in $x \in \binset^n$ is odd.
Similarly, $\MAJ_n \colon \binset^n \to \binset$ is the function such that $\MAJ_n(x) = 1$ if and only if
the number of ones in $x$ is at least $n/2$.

A \emph{distribution} $\mu$ on a finite set $X$ is a function $\mu \colon X \to [0, 1]$
such that $\sum_{x \in X} \mu(x) = 1$.
We write $x \sim \mu$ to indicate that $x$ is a random variable sampled from a distribution $\mu$.
Abusing notation,
we identify a finite set $X$ with the uniform distribution on $X$.
For example, $x \sim X$ means that $x$ is a uniform sample from $X$.

\headline{Organization}
The rest of this paper is organized as follows:
In Section \ref{sec:duality}, we prove our duality theorem
which links depth-3 formulas with one-sided approximation by depth-2 formulas.
In Section \ref{sec:hard}, we compute one of hardest distributions in some cases,
and prove Theorem \ref{thm:worst_to_average}.
In Sections \ref{sec:approxmonotone} and \ref{sec:universalbound},
Theorems \ref{thm:monotoneapprox_informal} and \ref{thm:universalapproximator_informal} are proved, respectively.
In Section \ref{sec:sat}, we study the one-sided-error correlation bound of the parity function.
We conclude with some open problems in Section~\ref{sec:conclusions}.

\section{Proof of Duality} \label{sec:duality} 
In this section,
we prove that the complexity of depth-3 formulas
is closely related to correlation bounds of 
CNF formulas with one-sided error.
As mentioned, our results hold for general settings
(\eg worst-case complexity of depth-$(d+1)$ formulas 
is almost equivalent to one-sided-error average-case complexity of depth-$d$ formulas);
however, for simplicity, we focus on the case of CNF formulas.
We first define correlation bounds of CNF formulas.
\begin{definition}
  Let $f \colon \binset^n \to \binset$ be a function.
  \begin{itemize}
    \item
      Let $\CNF{f}$ be the set of all CNF formulas $\phi$
      such that $\phi^{-1}(1) \subset f^{-1}(1)$
      (\ie $\phi$ does not err on inputs $x$ such that $f(x) = 0$).
    \item
      Let $\mu$ be a distribution on $f^{-1}(1)$.
      The \emph{maximum one-sided correlation per size}%
      \footnote{In order to make this definition well-defined,
        we regard $0 / 0$ as $0$.}
      $ \cor[\mu](f) $
      with $f$
      with respect to $\mu$
      is
      \[
        \cor[\mu](f) := \max_{\phi \in \CNF{f}} \Pr_{x \sample \mu} \left[ \phi(x) = 1 \right] \,/\, |\phi|.
      \]
    \item
      The \emph{maximum one-sided correlation per size} $\cor(f)$
      with $f$
      is
      $
        \cor(f) := \min_\mu \cor[\mu](f),$
      where the minimum is taken over all distribution $\mu$ on $f^{-1}(1)$.
    \item
      We denote by $\cor^+(f)$ the maximum one-sided correlation per size
      with $f$ for \emph{monotone} CNF formulas (\ie formulas without negated literals).
  \end{itemize}
\end{definition}

The following duality is the main principle that we will exploit.
\begin{theorem}
  \label{thm:correlation}
  Let $f \colon \binset^n \to \binset$ be a function
  such that $\formulathree(f) > 0$.
  The following holds:
  \[
    \left( \cor(f) \right)^{-1} \le
    \formulathree(f)
    \le
    (1 + n \cdot \ln 2) \cdot \left( \cor(f) \right)^{-1}.
  \]
%  In fact, the upper bound can be improved as
%  $ \formulathree(f)
%    \le
%    \left(1 + \ln |f^{-1}(1)| + \ln \cor(f) \right) \left( \cor(f) \right)^{-1}.$
\end{theorem}
\begin{proof}
  The idea is to regard the depth-3 formulas as
  a set cover instance, apply a
  linear programming relaxation
  and take the dual of the LP relaxation.

  Specifically, if $f$ is computed by a depth-3 formula,
  $f$ can be written as a disjunction of CNF formulas $\phi_1, \dots, \phi_m$
  (\ie $f = \bigvee_{i=1}^m \phi_i$, or equivalently, $f^{-1}(1) = \bigcup_{i=1}^m \phi_i^{-1}(1)$)
  such that $\phi_i^{-1}(1) \subset f^{-1}(1)$.
  Moreover, the size of the depth-3 formula is $\sum_{i=1}^m |\phi_i|$.
  Therefore, the problem of finding the minimum depth-3 formula computing $f$
  is equivalent to the following set cover instance:
  We want to cover
  a universe $U := f^{-1}(1)$
  by using a collection of
  sets $\phi^{-1}(1)$ such that
  $\phi$ is a CNF formula and $\phi^{-1}(1) \subset f^{-1}(1)$ (\ie $\phi \in \CNF{f}$),
  where the cost of $\phi^{-1}(1)$ is defined to be $|\phi|$.
  The minimum cost of this set cover instance is 
  exactly equal to
  the minimum size of depth-3 formulas computing $f$.

  Now we consider a linear programming relaxation of the set cover instance
  on variables $x_\phi$ for each $\phi \in \CNF{f}$:
  \begin{align*}
    \text{minimize} \quad &
    \sum_{\phi \in \CNF{f}} |\phi| \cdot x_\phi
    \\
    \text{subject to} \quad &
    \sum_{\phi \in \CNF{f}  \,:\ \phi(e) = 1} x_\phi \ge 1
    & \text{for all $e \in f^{-1}(1)$,}
    \\
    &
    x_\phi \ge 0
    & \text{for all $\phi \in \CNF{f}$.}
  \end{align*}
  Let $s^*$ denote the optimal value of this
  linear programming.
  Since it is a linear programming relaxation of the set cover problem,
  it holds that $s^* \le \formulathree(f)$.
  Moreover, it is well known that the integrality gap of 
  the set cover problem is at most
  $1 + \ln |U| \le 1 + n \cdot \ln 2$
  (Lov{\'a}sz~\cite{Lov75} and Chvatal~\cite{Chv79}; see also Vazirani~\cite{Vaz01}).
  Thus, we have $s^* \le \formulathree(f) \le (1 + n \cdot \ln 2) \cdot s^*$.

  It remains to claim that $s^* = (\cor(f))^{-1}$.
  By the strong duality of linear programming,
  $s^*$ is equal to the optimal value of 
  the following dual problem
  on variables $y_e$ for each $e \in f^{-1}(1)$:
  \begin{align*}
    \text{maximize} \quad &
    \sum_{e \in f^{-1}(1)} y_e
    \\
    \text{subject to} \quad &
    \sum_{e \in \phi^{-1}(1)} y_e \le |\phi|
    & \text{for all $\phi \in \CNF{f}$,}
    \\ &
    y_e \ge 0
    & \text{for all $e \in f^{-1}(1)$.}
  \end{align*}
  Let $v := \sum_{e \in f^{-1}(1)} y_e$.
  Define%
  \footnote{Since $0 < \formulathree(f) \le (1 + n \cdot \ln 2) \cdot s^*$, we have $s^* > 0$.
  Thus, we may assume, without loss of generality, that $v > 0$.}
  $\mu$ to be the
  distribution on $f^{-1}(1)$ such that $\mu(e) := y_e / v$ for each $e \in f^{-1}(1)$.
  Since $\sum_{e \in \phi^{-1}(1)} y_e = v \cdot \Pr_{e \sample \mu} \left[ \phi(e) = 1 \right]$,
  we can rewrite the dual problem as 
  the following optimization problem over all distributions $\mu$ on $f^{-1}(1)$:
  \begin{align*}
    \text{maximize} \quad &
    v
    \\
    \text{subject to} \quad &
    v \le 
    \left( \Pr_{e \sample \mu} \left[ \phi(e) = 1 \right] / |\phi| \right)^{-1}
    & \text{for all $\phi \in \CNF{f}$,}
  \end{align*}
  whose optimal value is clearly equal to $(\cor(f))^{-1}$.
  Hence, the optimal value $s^*$  of the dual problem is
  equal to $(\cor(f))^{-1}$.
\end{proof}

\section{Computing Hard Distributions} \label{sec:hard}

In this section, we give some examples
for which one of hard distributions can be determined.
If a function $f$ is ``symmetric'' in a certain sense,
one can reduce the number of variables, as in Yao's minimax principle \cite{Yao77}.
The following lemma gives such a general method.
\begin{lemma}
  \label{lemma:symmetric}
  Let $f \colon \binset^n \to \binset$.
  Let $\Pi$ be a subgroup 
  of the symmetric group on $\binset^n$
  such that,
  for each $\pi \in \Pi$ and each $\phi \in \CNF{f}$,
  there exists a formula $\phi_{\pi} \in \CNF{f}$
  such that $\phi = \phi_{\pi} \circ \pi$,
  $f = f \circ \pi$,
  and $|\phi_{\pi}| = |\phi|$.
  Then, 
  there exists a distribution $\mu$ on $f^{-1}(1)$ such that
  $\mu \circ \pi = \mu$ for any $\pi \in \Pi$ and
  $\cor(f) = \cor[\mu](f)$.
\end{lemma}
\begin{proof}
%  [Proof of Lemma \ref{lemma:symmetric}]
  Let $\mu^*$ be a ``hard'' distribution $\mu$ 
  such that $\cor[\mu^*](f) = \cor(f)$.
  Define a distribution $\mu$ on $f^{-1}(1)$ as
  $\mu(x) = \Exp_{\pi \usample \Pi} [ \mu^*(\pi(x)) ]$
  for each $x \in f^{-1}(1)$,
  where the probability is taken over the uniform distribution over $\Pi$.
  Since $\Pi$ is a subgroup of the symmetric group on $\binset^n$,
  for any $\pi \in \Pi$ and $x \in f^{-1}(1)$,
  we have $\mu(\pi(x)) =\Exp_{\pi' \usample \Pi} [ \mu^*(\pi'(\pi(x))) ] =
  \Exp_{\sigma \usample \Pi} [ \mu^* (\sigma(x)) ] = \mu(x)$,
  where,
  in the second equality,
  we replaced $\pi' \circ \pi$ by $\sigma$
  and used the fact that $\Pi$ is a group.

  By the definition of $\mu^*$, we have $\cor[\mu^*](f) \le \cor[\mu](f)$.
  Hence, it suffices to claim that
  $\cor[\mu](f) \le \cor[\mu^*](f)$.
  Indeed, for any $\phi \in \CNF{f}$,
  \begin{align*}
     \Pr_{x \sample \mu} \left[ \phi(x) = 1 \right]
     \quad &= \sum_{x \in f^{-1}(1)} \mu(x) \phi(x)
    \\ & =
     \sum_{x \in f^{-1}(1)} \Exp_{\pi \usample \Pi} [ \mu^*(\pi(x)) ] \phi(x)
     \explain{by the definition of $\mu$}
    \\ & =
     \sum_{y \colon \pi^{-1}(y) \in f^{-1}(1)} \Exp_{\pi \usample \Pi} [ \mu^*(y) ] \phi(\pi^{-1}(y))
     \explain{by defining $y := \pi(x)$}
    \\ & =
     \sum_{y \in f^{-1}(1)} \Exp_{\pi \usample \Pi} [ \mu^*(y) ] \phi_\pi(y)
     \explain{by $f \circ \pi = f$ and $\phi = \phi_\pi \circ \pi$}
    \\ & =
     \Exp_{\pi \usample \Pi} \left[ \sum_{y \in f^{-1}(1)} \mu^*(y) \phi_\pi(y) \right]
     \\ & \le
     \Exp_{\pi \usample \Pi} \left[ \cor[\mu^*](f) \cdot |\phi_\pi| \right]
     = \cor[\mu^*](f) \cdot |\phi| 
     \explain{by $|\phi| = |\phi_\pi|$}
  \end{align*}
  It follows that $\cor[\mu](f) = \max_{\phi \in \CNF{f}} \Pr_{x \sample \mu} \left[ \phi(x) = 1 \right] / |\phi| \le \cor[\mu^*](f)$.
\end{proof}

In particular,
if a function $f \colon \binset^n \to \binset$ is symmetric
in the sense that the value of $f(x)$ depends only on the number of ones in $x$,
then 
one of the hardest distributions $\mu$
is also symmetric.
\begin{corollary}
  \label{cor:symmetric}
  If a function $f \colon \binset^n \to \binset$ is symmetric,
  then there exists a symmetric distribution $\mu$ such that $\cor(f) = \cor[\mu](f)$.
\end{corollary}
\begin{proof}
  Let $\Pi$ be the set of all permutations
  that are induced by the $n!$ permutations on $(x_1, \cdots, x_n)$,
  and apply Lemma \ref{lemma:symmetric}.
\end{proof}

In the case of $\PAR_n$, we can completely determine one of
the hardest distributions.
\begin{corollary}
  \label{cor:hard_dist_of_parity}
  Let $\mu$ be the uniform distribution on $\PAR_n^{-1}(1)$.
  Then we have $\cor(\PAR_n) = \cor[\mu](\PAR_n)$.
\end{corollary}
\begin{proof}
  %[Proof of Corollary \ref{cor:hard_dist_of_parity}]
  Let $\Pi$ be the set of all permutations $\pi \colon \binset^n \to \binset^n$
  that negates an even number of coordinates
  (\eg $\pi(x) = \neg{x_1} \cdot \neg{x_2} \cdot x_3 \cdots x_n$).
  Note that any permutation $\pi \in \Pi$ does not change the parity.
  By Lemma \ref{lemma:symmetric}, there exists a 
  distribution on $f^{-1}(1)$ such that $\mu \circ \pi = \mu$
  for any $\pi \in \Pi$ and $\cor(f) = \cor[\mu](f)$.
  We claim that $\mu(x) = \mu(1 0^{n-1})$ for any $x \in \PAR_n^{-1}(1)$.
  Indeed, 
  it is easy to see that,
  for any input $x$ whose parity is $1$,
  there exists a permutation $\pi_x \in \Pi$ that maps $x$ to $1 0^{n-1}$.
  Since $\mu \circ \pi_x = \mu$, we have $\mu(x) = \mu(\pi_x(x)) = \mu(1 0^{n-1})$.
  Hence, $\mu$ is the uniform distribution on $\PAR_n^{-1}(1)$.
\end{proof}

More generally,
Corollary \ref{cor:hard_dist_of_parity} establishes
an equivalence 
between
depth-$(d+1)$ formula lower bounds
and
one-sided-error correlation bounds of depth-$d$ formulas with $\PAR_n$ on the uniform distribution.
In particular, we obtain one-sided-error correlation bounds for $\AC^0[p]$
as stated in Theorem \ref{thm:worst_to_average}.

\begin{proof}
  [Proof of Theorem \ref{thm:worst_to_average}]
  We apply the duality theorem (Theorem \ref{thm:correlation})
  to $\OR \circ \AC^0[p]$ (namely, the class of circuits that consist of a top OR gate fed by disjoint $\AC^0[p]$ circuits).
  Recall that $\cor^{\AC^0[p]}(f)$ denotes the one-sided-error correlation per size of $\AC^0[p]$
  circuits with the function $f$.
  Theorem \ref{thm:correlation} shows that 
  $\cor^{\AC^0[p]}(f)$ is at most $O(n)$ times the inverse of the size of $\OR \circ \AC^0[p]$ circuits for
  computing $f$.
  Now we apply Smolensky's lower bound \cite{Smo87} on the parity function
  to $\OR \circ \AC^0[p]$ circuits, and obtain
  $\Pr[C(x) = 1] / |C| = \cor^{\AC^0[p]}(\PAR_n) \le (1 + n \cdot \ln 2) \cdot 2^{-n^{\Omega(1/(d+1))}} = 2^{-n^{\Omega(1/d)}}$
  for any $\AC^0[p]$ circuit $C$ of depth $d$.
  Here, by Corollary \ref{cor:hard_dist_of_parity},
  we may assume that the probability is taken over the uniform distribution $\mu$
  on $\PAR_n^{-1}(1)$.
  Therefore,
  \begin{align*}
    & \Pr_{x \sim \binset^n} [C(x) = \PAR_n(x)]
    \\ & = \frac{1}{2} + \frac{1}{2} \Pr_{x \sim \mu} [C(x) = 1]
    \explain{since $C(x) = 0$ if $\PAR_n(x) = 0$}
    \\ & \le \frac{1}{2} + \frac{1}{2} \cdot |C| \cdot 2^{-n^{\Omega(1/d)}}.
  \end{align*}
\end{proof}

\newcommand{\slice}[2]{S^{#1}_{#2}}
We can also compute one of the hardest distributions
for computing the majority function by \emph{monotone} formulas.
We write $\slice{n}{k}$ for the set of all inputs $x \in \binset^n$
such that the number of ones in $x$ is $k$.
%Abusing notation,
%we write $\slice{n}{k}$ for the set of all inputs $x \in \binset^n$
%such that the number of ones in $x$ is $k$.
\begin{proposition}
  \label{prob:hardest_for_maj}
  Let $\mu$ be the uniform distribution on $\slice{n}{\ceil{n/2}}$.
  Then, $\cor^+(\MAJ_n) = \cor[\mu]^+(\MAJ_n)$.
\end{proposition}
\begin{proof}
  %[Proof of Proposition \ref{prob:hardest_for_maj}]
  As observed in \cite{BHST14},
  for any monotone function $f$,
  it holds that $\Pr_{x \usample \slice{n}{k}} [ f(x) = 1 ] \le \Pr_{x \usample \slice{n}{k+1}} [ f(x) = 1]$
  for any $k < n$.
  Indeed, by double counting,
  \begin{align*}
    \# \ISet{ x \in \slice{n}{k} } { f(x) = 1 } \cdot (n - k)
    & = \# \ISet { (x, y) \in \slice{n}{k} \times \slice{n}{k+1} } { f(x) = 1, x \le y }
    \\ & \le \# \ISet { (x, y) \in \slice{n}{k} \times \slice{n}{k+1} } { f(y) = 1, x \le y }
    \\ & = \# \ISet{ y \in \slice{n}{k+1} } { f(y) = 1 } \cdot (k + 1),
  \end{align*}
  and hence 
  $\Pr_{x \usample \slice{n}{k}} [ f(x) = 1 ] \le \Pr_{x \usample \slice{n}{k+1}} [ f(x) = 1]$
  (by noting that $\frac{n-k}{k+1} \binom{n}{k} = \binom{n}{k+1}$).
  Therefore, the uniform distribution on $\slice{n}{k}$
  is ``harder'' than $\slice{n}{k+1}$ for monotone formulas.
  Details follow.

  Since $\MAJ_n$ is a symmetric function, 
  by 
  Corollary \ref{cor:symmetric}
  there exists a symmetric distribution $\mu^*$ such that
  $\cor^+(\MAJ_n) = \cor[\mu^*]^+(\MAJ_n)$.
  (Note that, while Corollary \ref{cor:symmetric}
  is stated for nonmonotone formulas, the same proof can be applied
  to the case of monotone formulas.)
  Let $\phi$ be a monotone CNF formula that achieves
  the maximum one-sided correlation per size $\cor[\mu^*]^+(\MAJ_n)$.
  Then,
  \begin{align*}
  \cor[\mu^*]^+(\MAJ_n) \cdot |\phi|
  & = \Pr_{x \sample \mu^*} [ \phi(x) = 1 ]
  \\ & = \sum_{k = \ceil{n/2}}^n
  \Pr_{x \sim \mu^*} [x \in \slice{n}{k}]
  \cdot \Pr_{x \usample \slice{n}{k}} [ \phi(x) = 1 ]
  \explain{since $\mu^*$ is symmetric}
  \\ & \ge \sum_{k = \ceil{n/2}}^n
  \Pr_{x \sim \mu^*} [x \in \slice{n}{k}]
  \cdot \Pr_{x \usample \slice{n}{\ceil{n/2}}} [\phi(x) = 1] 
  \explain{since $\phi$ is monotone}
  \\ & = \Pr_{x \sample \mu} [ \phi(x) = 1 ].
  \end{align*}
  Hence, we have $\cor^+(\MAJ_n) \le \cor[\mu]^+(\MAJ_n) \le \cor[\mu^*]^+(\MAJ_n) = \cor^+(\MAJ_n)$.
\end{proof}

We can state the one-sided correlation per size $\cor^+(\MAJ_n)$
in terms of hypergraphs.
Specifically, the following holds.
\begin{corollary}
  \label{cor:characterizingmajority}
  For a hypergraph $\mathcal F$,
  let
  $|\mathcal F|$ denote the number of edges in $\mathcal F$,
  $\tau(\mathcal F)$ denote the minimum size of hitting sets of $\mathcal F$,
  and $t(\mathcal F)$ denote the number of minimum hitting sets of $\mathcal F$.
  Let $T(n, \tau) := \max_{\mathcal F \colon \tau(\mathcal F) = \tau} t(\mathcal F) / |\mathcal F|$,
  where the maximum is taken over all hypergraphs of $n$ vertices.
  Let $\formulathreemonotone(\MAJ_n)$ be the minimum depth-3 monotone formula size for computing $\MAJ_n$.
  Then, the following holds:
  \[
    \frac{1}{ T(n, \ceil{n/2})}
    \le \frac{\formulathreemonotone(\MAJ_n)} { \binom{n}{\ceil{n/2}} }
    \le \frac{1 + \ln 2 \cdot n}{ T(n, \ceil{n/2}) }.
\]
\end{corollary}
\begin{proof}
  %[Proof of Corollary \ref{cor:characterizingmajority}]
  Given a monotone CNF formula $\phi$,
  we define a hypergraph
  $\mathcal F$ so that
  each edge of $\mathcal F$
  is the set of literals in a clause of $\phi$.
  (Note that, since $\phi$ is a monotone CNF,
  every literal is a positive.)
  We can naturally identify $x \in \binset^n$ with $x \subset \NumSet{n}$.
  It is easy to see that
  $\phi$ accepts $x \in \binset^n$ if and only if $x \subset \NumSet{n}$ is
  a hitting set of $\mathcal F$.
  Therefore, 
  the constraint that $\phi^{-1}(1) \subset \MAJ_n^{-1}(1)$
  corresponds to the constraint that $\tau(\mathcal F) = \ceil{n/2}$.
  The number $|\mathcal F|$ of edges in $\mathcal F$ is equal to $|\phi|$.
  Moreover, 
  $t(\mathcal F) = \binom{n}{\ceil{n/2}} \cdot \Pr_{x \usample \slice{n}{\ceil{n/2}}} [\phi(x) = 1]$.
  Therefore, $T(n, \ceil{n/2}) = \binom{n}{\ceil{n/2}} \cor^+(\MAJ_n)$
  and the result follows from Theorem \ref{thm:correlation}.
\end{proof}

In the case of graphs (instead of hypergraphs),
the extremal structure
that
maximizes the number of hitting sets
under the constraint
that
the minimum size of
hitting sets is bounded from below
is well understood.
Thus, in the case of graphs, 
or equivalently,
in the case of bottom fan-in 2,
we are able to determine the minimum depth-3 circuit size
for computing the majority function.
\begin{proposition}
  \label{prop:monotonemaj}
  The minimum size of
  depth-3 monotone circuits
  with bottom fan-in 2
  for computing $\MAJ_n$
  is $\widetilde \Theta (2^{n / 2})$.
\end{proposition}
\begin{proof}
  %[Proof of Proposition \ref{prop:monotonemaj}]
  It is well known that 
  the extremal structure
  that maximizes
  the number of hitting sets (of each cardinality)
  given the constraint that $\tau(\mathcal F) = \tau$
  is the complement of Tur\'an's graph \cite{Tur41}
  (\ie a disjoint union of cliques of almost equal size;
  hence, in our case, it is a disjoint union of $n/2$ edges).
  (This fact was independently proved by many people,
  such as 
  Zykov \cite{Zyk49},
  Erd{\H{o}}s \cite{Erd62}, Sauer \cite{Sau71},
  Had{\v{z}}iivanov \cite{Had76}, and Roman \cite{Rom76}.
  For a recent simple proof, see Cutler and Radcliffe \cite{CR11}.)

  The number of minimum hitting sets is $O(2^{n / 2})$.
  Since the number of edges in a graph of $n$ vertices is at most $n^2 / 2$,
  the maximum correlation per size is $\widetilde\Theta(2^{- n / 2})$.
  Moreover, depth-3 circuit size and depth-3 formula
  with bottom fan-in 2 are the same up to a factor of $O(n^2)$ (\eg see \cite[Proposition 2.2]{PSZ00}).
  Thus, the minimum circuit size is equal to $2^{n / 2}$
  up to a polynomial factor.
\end{proof}

We remark that the same quantitative bound can be obtained by using the
techniques of \cite{HJP95}.
Proposition \ref{prop:monotonemaj} merely gives its alternative proof
from the techniques of the different literature.

\section{Approximating Monotone Functions by Depth-2 Formulas}
\label{sec:approxmonotone}

It is reasonable to conjecture that an analog of Proposition \ref{prob:hardest_for_maj} holds for
\emph{nonmonotone} CNF formulas: that is, we conjecture that the uniform distribution on $\slice{n}{\ceil{n/2}}$
is close to the hardest distribution of CNF formulas approximating $\MAJ_n$.
And it is tempting to guess that, since $\MAJ_n$ is a monotone function,
an optimal CNF formula approximating $\MAJ_n$ should be a monotone CNF formula.
This is true when \emph{exact} complexity of computing monotone function is concerned, by Quine's theorem \cite{Qui54}.
It is, however, not true when approximation is concerned:
Blais, H{\aa}stad, Servedio and Tan \cite{BHST14} showed that 
there is a monotone function on $n$ variables that can be approximated by some CNF formula of size $O(n)$,
but cannot be approximated by any monotone CNF formula of size less than $n^2$.

Here we significantly improve their bounds:
\begin{theorem}
  [Theorem \ref{thm:monotoneapprox_informal}, restated formally]
  \label{thm:monotoneapprox}
  There exists a function $\epsilon \colon \Nat \to (0, 1)$ and a family of monotone Boolean functions $f = \{ f_n \colon \binset^n \to \binset \}_{n \in \Nat}$ and 
  a family of distributions $\{ \mu_n \}_{n \in \Nat}$ on $f_n^{-1}(1)$ such that
  \begin{itemize}
    \item
      there is some CNF formula $\{ \phi_n \}_{n \in \Nat}$ of size at most $n^{O(1)}$ satisfying $\phi_n^{-1}(1) \subset f_n^{-1}(1)$ and 
      $\Pr_{x \sim \mu_n} [\phi_n(x) = 1] \le \epsilon(n)$, whereas
    \item
      any monotone CNF formula $\{ \phi_n^+ \}_{n \in \Nat}$ satisfying $(\phi_n^+)^{-1}(1) \subset f_n^{-1}(1)$ and 
      $\Pr_{x \sim \mu_n} [\phi_n^+(x) = 1] \le \epsilon(n)$
      requires size $2^{n^{\Omega(1)}}$.
 \end{itemize}
\end{theorem}
In fact, this theorem is an immediate consequence of our duality theorem together with
recently improved Ajtai-Gurevich's theorem \cite{AG87}.
\begin{theorem}
  [Chen, Oliveira and Servedio \cite{COS15}]
  \label{thm:COS15}
  There exists a family of monotone functions $f = \{ f_n \colon \binset^n \to \binset \}_{n \in \Nat}$
  such that
  \begin{enumerate}
    \item
      $\formulathree(f_n) = n^{O(1)}$, and
    \item
      any monotone depth-$d$ circuit%
      \footnote{In fact, their result gives a lower bound for monotone \emph{majority} circuits.
      For our purpose, it is sufficient to use a lower bound for monotone $\textrm{OR} \circ \textrm{AND} \circ \textrm{OR}$ circuits.}
      computing $f$
      requires size $2^{n^{\Omega(1/d)}}$.
  \end{enumerate}
\end{theorem}
Using our duality theorem,
we now transfer Theorem \ref{thm:COS15} about depth-3 circuits to
Theorem \ref{thm:monotoneapprox} about approximation by depth-2 formulas.
\begin{proof}
  [Proof of Theorem \ref{thm:monotoneapprox}]
  By the second item of Theorem \ref{thm:COS15},
  $\formulathreemonotone(f_n) \ge 2^{n^{\Omega(1)}}$.
  By Theorem \ref{thm:correlation}, $\cor^+(f_n) \le O(n) / \formulathreemonotone(f_n) = 2^{- n^{\Omega(1)}}$.
  That is, there exists a distribution $\mu_n$ on $f_n^{-1}(1)$ such that
  \[
    \Pr_{x \sim \mu_n} [\phi^+(x) = 1] \le |\phi^+| \cdot 2^{-n^{\Omega(1)}}
  \]
  for any monotone CNF formula $\phi^+$.
  In particular, for any monotone CNF formula $\phi^+$ of size less than $2^{\frac{1}{2} \cdot n^{\Omega(1)}}$,
  we have $\Pr_{x \sim \mu_n} [\phi^+(x) = 1] \le 2^{- \frac{1}{2} n^{\Omega(1)} } = 2^{- n^{\Omega(1)}}$.

  On the other hand, by Theorem \ref{thm:correlation} and the first item of Theorem \ref{thm:COS15},
  we obtain $\cor(f) \ge 1 / \formulathree(f) \ge n^{-O(1)}$.
  Hence, for the distribution $\mu_n$ defined above, there exists a CNF formula $\phi_n$ such that
  \begin{align}
    \label{eq:correlationlowerbound}
    \Pr_{x \sim \mu_n} [ \phi_n(x) = 1 ] \ge |\phi_n| \cdot n^{-O(1)} \ge n^{-O(1)}
  \end{align}
  since $|\phi_n| \ge 1$.
  Define $\epsilon(n) = n^{-O(1)}$ to be the rightmost lower bound in \eqref{eq:correlationlowerbound}.
\end{proof}

\section{Universal Bounds on Approximation by CNFs}
\label{sec:universalbound}
In this section, we prove a tight bound on the maximum size of CNF formulas approximating a function with one-sided error.
We say that a CNF formula $\phi$ approximates a function $f$
with (one-sided error and) advantage $\epsilon$ if $\Pr_{x \sim f^{-1}(1)} [\phi(x) = 1] \ge \epsilon$ and $\lessthan \phi f$.
%We restate Theorem \ref{thm:universalapproximator_informal}
\begin{theorem}
  [Theorem \ref{thm:universalapproximator_informal}, restated formally]
  \label{thm:universalapproximator}
  For all large $n \in \Nat$,
  the following holds.
  \begin{itemize}
    \item
      For any function $f \colon \binset^n \to \binset$,
    there exists some $\epsilon \in (0, 1]$ and some CNF formula
      of size at most $\epsilon2^{n+3} / n$
      approximating $f$ with one-sided error and advantage $\epsilon$.
    \item
      There exists a function $f \colon \binset^n \to \binset$
      such that,
    for any $\epsilon \in (0,1]$,
      any CNF formula approximating $f$ with one-sided error and advantage $\epsilon$
      must have size at least $\epsilon 2^{n-7} / n$.
  \end{itemize}
\end{theorem}

In order to obtain the upper bound, we prove the following lemma.
%We first prove the second item, which can be proved via a simple counting argument.
%Intuitively, this follows from the fact%
%\footnote{Note that we do not count the number of bottom gates in depth-3 formulas.}
%that $\formulathree(f) = \Omega(2^n/n)$ for a random function $f$
%(and the duality theorem).
%Indeed, the number of depth-3 formulas of size at most $s$ can be bounded above by
%$2^{s \cdot O(\log s + n)}$,
%because any depth-3 formula can be described by
%the $i$th clause and the index of the depth-2 formula which contains the $i$th clause, for each $i \in \{1, \cdots, s\}$.
%However, since the duality theorem may lose a factor of $n$,
%we cannot transfer this counting argument to approximation by depth-2.
%We thus directly prove the second item.
\begin{lemma}
  \label{lemma:universalbound}
  $\formulathree(f) \le 2^{n+3} / n$ for any function $f$.
\end{lemma}
A classical theorem by Lupanov \cite{Lup65} states that,
for any function $f$,
the number of literals in the smallest depth-3 formula computing $f$ is at most $O(2^n / \log n)$ (and this is tight).
In contrast, Lemma \ref{lemma:universalbound} deals with the number of 
AND gates at the middle layer.
To the best of our knowledge, for this size measure, the universal upper bound was not studied before;
however, Lupanov's construction can be adapted to our size measure as well, by changing some parameters in the construction.
%however, our construction is in fact a degenerate case of the construction of Lupanov
%with different parameters.
\begin{proof}
  [Proof of Lemma \ref{lemma:universalbound}]
  The main idea is to cover the whole space by a sphere of diameter 1,
  and then to describe the function inside each sphere by a relatively small CNF formula.
  Take a parameter $D = 2^d \supplement{ < n}$ (which is fixed later) for some $d \in \Nat$.
  A sphere of diameter 1 with center $a$ is denoted by $S_a \subset \binset^D$.
  That is, $S_a$ consists of all the strings $y \in \binset^D$ such that
  $y$ and $a$ differ in exactly one coordinate.
  We use several simple facts from \cite{Lup65}.
  \begin{fact}
    \label{fact:coverbysphere}
    There is some subset $A_D \subset \binset^D$ such that $\ISet{S_a}{a \in A_D}$ partitions $\binset^n$
    %(\ie $\coprod_{a \in A_D} S_a = \binset^n$).
    (\ie $\coprod_{a \in A_D} S_a = \binset^n$).
  \end{fact}
  \begin{proof}
    [Proof Sketch]
    Let $H \in \GF2^{d \times 2^d}$ be a $d \times 2^d$ matrix whose $i$th column is
    the $d$-bit binary representation of $i \in \NumSet{2^d}$.
    Define $A_D$ to be the kernel of a linear map $H$ regarded as $H \colon \GF2^{2^d} \to \GF2^d$.
    The fact follows from the properties of the Hamming code \cite{Ham50}.
  \end{proof}
  \begin{fact}
    For any $a \in \binset^D$,
    the characteristic function of $S_a$ can be computed by a CNF formula $\phi_a$ of size at most $D^2$.
  \end{fact}
  \begin{proof}
    [Proof Sketch]
    $x \in S_a$ if and only if
    (1) there is some $i \in \NumSet{D}$ such that $x_i \neq a_i$
    and (2) for any pair of $i < j$, $x_i = a_i$ or $x_j = a_j$.
  \end{proof}
  The next fact states that, if restricted to a sphere $S_a$, any function $g$ can be described
  by a single clause $C_a^g$.
  \begin{fact}
    \label{fact:clauseinsphere}
    For any function $g \colon \binset^D \to \binset$ and $a \in \binset^D$, there exists a clause $C_a^g$ such that
    $g(y) = C_a^g(y)$ for any $y \in S_a$.
%    $\phi_a(y) \land g(y) = \phi_a(y) \land C_g(y)$.
  \end{fact}
  \begin{proof}
    [Proof Sketch]
    For $b \in \binset$, let $y_i^{b}$ denote a positive literal $y_i$ if $b=1$
    and a negative literal $\neg y_i$ otherwise.
    Define $C_a^g(y) := \bigvee_{i} y_i^{1 - a_i}$ where
    the disjunction is taken over all $i \in \NumSet{D}$ such that 
    $g$ evaluates to $1$ if the $i$th coordinate of $a$ is flipped
    (\ie $(a_1, \dots, 1 - a_i, \dots a_D) \in g^{-1}(1)$).
  \end{proof}
  Using these facts, we can now describe any function $f \colon \binset^n \to \binset$ by a small depth-3 formula.
  Regard $\binset^n = \binset^D \times \binset^{n-D}$.
  For $y \in \binset^D$ and $z \in \binset^{n-D}$,
  \begin{align*}
    f(y, z)
    & = f(y, z) \land \left( \bigvee_{a \in A_D} \phi_a(y) \right)
    \explain{by Fact \ref{fact:coverbysphere}}
    \\ & = \bigvee_{a \in A_D} \phi_a(y) \land f(y, z)
    \\ & = \bigvee_{a \in A_D} \left\{ \phi_a(y) \land \bigwedge_{w \in \binset^{n-D}} \left(  f(y, w) \lor \bigvee_{j\in\NumSet{n-D}} z_j^{1-w_j} \right) \right\}
    \\ & = \bigvee_{a \in A_D} \left\{ \phi_a(y) \land \bigwedge_{w \in \binset^{n-D}} \left(  C_a^{f(\blank, w)}(y) \lor \bigvee_{j\in\NumSet{n-D}} z_j^{1-w_j} \right) \right\}.
    \explain{by Fact \ref{fact:clauseinsphere}}
  \end{align*}
  This is a depth-3 formula of size at most
  $|A_D| \cdot (|\phi_a| + 2^{n-D}) \le 2^D D + 2^n / D$.
  Define $d := \floor{\log (n/2)}$.
  Then, $n / 4 \le D = 2^d \le n / 2$ and hence
  the size is at most $2^{n/2} \cdot n/2 + 4 \cdot 2^n / n \le 2^{n+3} / n$.
  This completes the proof of Lemma \ref{lemma:universalbound}.
\end{proof}

\begin{remark}
  Lemma \ref{lemma:universalbound} is tight up to a constant factor.
  Indeed, a simple counting shows that
  the number of 
  depth-3 formulas of size $s$
  is 
  at most $2^{O(s \log s + s n)}$,
  which is much less than $2^{2^n}$ for $s \ll 2^n / n$.
\end{remark}

\begin{proof}
  [Proof of Theorem \ref{thm:universalapproximator}]
  \begin{enumerate}
    \item
      By Lemma \ref{lemma:universalbound} and Theorem \ref{thm:correlation},
      we obtain $(\cor(f))^{-1} \le \formulathree(f) \le 2^{n+3} / n$.
      Hence, for any distribution $\mu$ on $f^{-1}(1)$ (and in particular the uniform distribution on $f^{-1}(1)$),
      there exists some CNF formula $\phi$ such that
      \[
        |\phi| \le 2^{n+3} / n \cdot \Pr_{x \sim \mu} [\phi(x) = 1].
      \]
      Therefore,
      $\phi$ approximates $f$ with advantage $\epsilon$ for $\epsilon := \Pr_{x \sim \mu} [\phi(x) = 1]$,
      and $|\phi| \le \epsilon \cdot 2^{n+3} / n$.
    \item
      Let $f \colon \binset^n \to \binset$ be a random function.
      That is, for each input $x \in \binset^n$, pick $f(x) \sim \binset$ uniformly at random and independently.
      Fix any CNF formula $\phi$ and advantage $\epsilon$.
      We will bound the probability that $\phi$ approximates $f$ with one-sided error and advantage $\epsilon$.
      \begin{claim}
        $ \Pr_f \left[ \text{ $\lessthan{\phi}{f}$ and $\Pr_{x \sim f^{-1}(1)} [\phi(x) = 1] \ge \epsilon$ } \right] \le 2^{-\epsilon 2^{n-4}}$
      \end{claim}
      By definition, $\Pr_f [\lessthan{\phi}{f}] = 2^{- |\phi^{-1}(1)|}$.
      This probability is bounded above by $2^{ - \epsilon 2^{n-2}}$ if $|\phi^{-1}(1)| > \epsilon 2^{n-2}$,
      in which case the claim holds.

      It remains to consider the case when $|\phi^{-1}(1)| \le \epsilon 2^{n-2}$.
      First, note that $\Pr_{x \sim f^{-1}(1)} [\phi(x) = 1] \ge \epsilon$ is equivalent to $|\phi^{-1}(1)| \ge \epsilon |f^{-1}(1)|$
      under the assumption that $\lessthan{\phi}{f}$.
      Therefore, the probability in the claim can be bounded above by
      \[
        \Pr\left[ \epsilon |f^{-1}(1)| \le |\phi^{-1}(1)| \right] \le
        \Pr\left[ |f^{-1}(1)| \le 2^{n-2} \right] \le 2^{-2^{n-4}} \le 2^{-\epsilon 2^{n-4}},
      \]
      where the second last inequality follows from the Chernoff bound.
      This completes the proof of the claim.

      Fix any size $s \in \Nat$.
      Since there are at most $3^{ns}$ CNF formulas of size $s$,
      \[
        \Pr_f \left[ \text{ $\exists\, \phi$ of size $s$ such that $\lessthan{\phi}{f}$ and $\Pr_{x \sim f^{-1}(1)} [\phi(x) = 1] \ge \epsilon$  } \right]
        \le 3^{ns} \cdot 2^{-\epsilon 2^{n-4}}.
      \]
      by the union bound.
      Define $s_\epsilon := \epsilon 2^{n-6} / n$.
      Then, for any fixed $\epsilon$,
      the probability that there exists a CNF formula $\phi$ of size $s_\epsilon$
      approximating $f$ with advantage $\epsilon$
      is at most $3^{ns_\epsilon} \cdot 2^{-\epsilon 2^{n-4}} \le 2^{\epsilon 2^{n-5} - \epsilon 2^{n-4}} = 2^{- \epsilon 2^{n-5}}$.

      Let $\mathcal E :=  \ESet { 2^{-n+5}, 2^{-n+6}, \cdots, 2^{0} }$.
      By the union bound over all $\epsilon \in \mathcal E$,
      \begin{align*}
        & \Pr_f \left[ \text{ $\exists\,\epsilon \in \mathcal E,\  \exists\, \phi $ of size $s_\epsilon$ such that $\phi$ approximates $f$ with advantage $\epsilon$ } \right]
        \\ & \le \sum_{i=0}^{n-5} 2^{-2^{i}} < \sum_{i=1}^{\infty} 2^{-i} = 1.
      \end{align*}
      Hence, there exists a function $f$ such that $f$ cannot be approximated
      by any CNF formula of size $s_\epsilon$ with any advantage $\epsilon \in \mathcal E$.
      While this gives us inapproximability for discrete values of $\epsilon \in \mathcal E$,
    we can extend it to an arbitrary advantage $\epsilon \in (0, 1]$ as follows:
      Given an arbitrary $\epsilon \ge 2^{-n+5}$, we take $\epsilon' \in \mathcal E$ such that $\epsilon / 2 \le \epsilon' \le \epsilon$.
      Then, $f$ cannot be approximated by any CNF formula of size $s_\epsilon / 2 \supplement{\le s_{\epsilon'}}$ with any advantage $\epsilon \supplement{\ge \epsilon'}$.
      On the other hand, if $\epsilon < 2^{-n+5}$ then $s_\epsilon < 1$ and hence
      $f$ cannot be approximated by any CNF formula of size at most $s_\epsilon$ (\ie a constant formula)
      with positive advantage,
      as we may assume that $f$ is not constant.
      Therefore, in any case,
    $f$ cannot be approximated by any CNF formula of size $s_\epsilon / 2 = \epsilon 2^{n-7}/n$ with any advantage $\epsilon \in (0, 1]$.
  \end{enumerate}
\end{proof}

\section{Satisfiability Coding Lemma With Width Reduction} \label{sec:sat}
In this section, we modify the satisfiability coding lemma \cite{PPZ99} and
obtain an upper bound on the number of isolated solutions of
a CNF formula.  Here, we say that an assignment $x \in \binset^n$
is an \emph{isolated solution} of a function $\phi$
if $\phi(x) = 1$ and $\phi(y) = 0$ for any adjacent assignment $y$ of $x$
(\ie $x$ and $y$ differ on exactly one coordinate).
We note that an upper bound on the number of isolated solutions
immediately implies a one-sided correlation bound of $\PAR_n$.

Paturi, Pudl{\'{a}}k, and Zane \cite{PPZ99}
developed the satisfiability coding lemma,
which states that
an isolated solution has a short description
(and thus the number of isolated solutions is small).
We say that a randomized algorithm $E$ is a \emph{randomized prefix-free encoding}
if, for any fixed randomness $r$ of $E$, the image of the algorithm $E_r$ that uses $r$ as randomness
is prefix-free (\ie no two strings in the image of $E_r$ contain the other as a prefix).
\begin{lemma}
  [Satisfiability Coding Lemma~\cite{PPZ99}]
  \label{lemma:sat}
  Let $\phi$ be any $k$-CNF formula on $n$ variables,
  and let $T \subset \binset^n$ be the set of isolated solutions of $\phi$.
  Then, 
  there exists a randomized prefix-free encoding $E(\blank; \phi) \colon T \to \binset^*$
  such that $\Exp [ E(x; \phi) ] \le n - n / k$ for any $x \in T$,
  where the expectation is taken over the coin flips
  of $E$.
  In particular, $|T| \le 2^{n - n / k}$.
\end{lemma}
In order to derive a depth-3 formula lower bound
of $\PAR_n$,
we need 
an upper bound in terms of the formula size $|\phi|$
instead of the width $k$.
In the context of satisfiability algorithms (\ie decoding algorithms of satisfying assignments),
Schuler~\cite{Sch05} gave a variant of the PPZ algorithm \cite{PPZ99}
that runs in time $2^{n - n / (\log |\phi| + 1)} \poly(n)$ for a CNF formula $\phi$
(instead of the running time $2^{n - n / k} \poly(n)$ of the PPZ algorithm
where $k$ is the width of $\phi$).
We note that Calabro, Impagliazzo and Paturi \cite{CIP06} 
gave another analysis of Schuler's width reduction technique.
Their analysis gives a satisfiability algorithm that runs in
time $2^{n - n / O(\log (|\phi|/n))} \poly(n)$ (see \cite{DH08}).
However, it seems that their analysis does not improve our depth-3 formula lower bound on
the parity function because the ``$O$'' notation hides some factor.
We thus 
incorporate Schuler's width reduction technique
into the satisfiability coding lemma, and obtain the following:
\begin{theorem}
  \label{thm:isolated}
  Let $\phi$ be a CNF formula of size at most $2^s$ on $n$ variables.
  Then the number of isolated solutions of $\phi$
  is at most $2^{n - n / (s + 2) + 1}$.
\end{theorem}
\begin{proof} %[Proof of Theorem \ref{thm:isolated}]
  We will construct a randomized prefix-free encoding of the isolated solution of $\phi$.
  Let $T$ be the set of isolated solutions of $\phi$.
  Let $k := s + 2$.

  The idea is to cut off long clauses of length greater than $k$ in $\phi$
  and then apply the satisfiability coding lemma.
  We will define
  $\widetilde \phi$
  to be a formula produced by cutting off long clauses in $\phi$
  so that the width of $\widetilde \phi$ is at most $k$.
  The satisfiability coding lemma yields
  a short encoding whose average length is $n - n / k$.
  However, since we cut off long clauses in $\phi$,
  it is possible that
  a satisfying assignment $x$ of $\phi$
  is not necessarily a satisfying assignment of $\widetilde \phi$.
  In this case, we specify the index of the clause in $\widetilde \phi$
  that is not satisfied by $x$.
  While it costs us additional $s$ bits on the code length,
  we can obtain the $k$ bits of information about $x$
  from the fact that $x$ does not satisfy the clause.

  To be more precise, 
  the definition of $\widetilde \phi$ is as follows:
  For each clause $C = l_1 \lor l_2 \lor \dots \lor l_m$,
  define $\widetilde C$ as $\widetilde C := l_1 \lor l_2 \lor \dots \lor l_k$ if $m > k$,
  and $\widetilde C := C$ otherwise.
  (We fix an arbitrary ordering of literals so that
  $\widetilde C$ is uniquely determined.)
  Then,
  for a CNF formula $\phi = \bigwedge_{i=1}^{|\phi|} C_i$,
  let us define a $k$-CNF formula $\widetilde \phi$
  as $\widetilde \phi := \bigwedge_{i=1}^{|\phi|} \widetilde C_i$.

  Now we describe the prefix-free encoding algorithm $\widetilde E(x; \phi)$
  for encoding an isolated solution $x$ of $\phi$.
  If $\widetilde \phi(x) = 1$,
  then apply the satisfiability coding lemma, that is, output $0 \cdot E(x; \phi)$
  (``0'' is a marker indicating the first case).
  Otherwise, 
  there exists a clause $\widetilde C$ in $\widetilde \phi$
  such that $\widetilde C(x) = 0$.
  Let $\widetilde C = \bigvee_{i=1}^k l_i$ be the lexicographically first clause in $\widetilde\phi$
  such that $\widetilde C(x) = 0$.
  (Note that the width of $\widetilde C$ is exactly equal to $k$
  because $\widetilde C \neq C$.)
  Let $i \in \binset^s$ be the binary encoding of the index of $\widetilde C$ in $\widetilde \phi$.
  Let $\phi' := \restr{\phi}{l_1=0, \dots, l_k = 0}$
  (\ie assign the variables in $\widetilde C$ so that $\widetilde C$ is falsified),
  and let $x' \in \binset^{n-k}$ be the part of the assignment $x$ other than
  variables in $l_1, \dots, l_k$.
  Output $1 \cdot i \cdot \widetilde E(x'; \phi')$.
  (That is, recursively call this procedure on input $x'$ and $\phi'$.)

  We claim that
  $\widetilde E(\blank; \phi) \colon T \to \binset^*$ is a (randomized) prefix-free encoding.
  To this end, we describe a decoding algorithm $\widetilde D(z; \phi)$:
  Given an input $z = \widetilde E(x; \phi)$,
  read the first bit $b \in \binset$ of $z$.
  If $b = 0$, then use the decoding algorithm for $E(\blank, \phi)$.
  Otherwise,
  read the next $s$ bits $i \in \binset^s$ of $z$
  so that $z = 1 \cdot i \cdot z'$,
  let $\widetilde C$ denote the $i$th clause in $\widetilde \phi$,
  define $\phi' := \restr{\phi}{\widetilde C = 0}$ as in the encoding algorithm of $\widetilde E$,
  and recursively call $\widetilde D(z'; \phi')$.
  It is easy to see that
  $\widetilde D(\widetilde E(x; \phi); \phi) = x$
  for any $x \in T$,
  and thus $\widetilde E$ is a prefix-free encoding.

  Finally, we analyze the average code length of the prefix-free encoding $\widetilde E$.
  We claim, by induction on the number of recursive calls of $\widetilde E$, % on the number $n$ of variables in $\phi$,
  that $\Exp \left[ \widetilde E(x; \phi) \right] \le n - n/k + 1$
  for any isolated solution $x$ of $\phi$.
  Here, the expectation is taken over the randomness of $\widetilde E$ (\ie the randomness of $E$).
  %Fix an arbitrary isolated solution $x$ of $\phi$.
  If $\widetilde \phi(x) = 1$ (\ie there is no recursive call),
  then $|\widetilde E(x; \phi)| = 1 + |E(x; \phi)|$,
  and thus $\Exp \left[ |\widetilde E(x; \phi)| \right] \le 1 + n - n / k$
  by Lemma \ref{lemma:sat}.
  Otherwise,
  we have
  $|\widetilde E(x; \phi)| = 1 + s + |\widetilde E(x'; \phi')|$,
  where $x'$ and $\phi'$ are defined as in the encoding algorithm.
  Since $x'$ is an isolated solution of $\phi'$,
  we can apply the induction hypothesis for $x' \in \binset^{n-k}$ and $\phi'$.
  Hence,
  \begin{align*}
    \Exp[|\widetilde E(x; \phi)|] &\le 1 + s + (n - k) - (n - k) / k + 1
    \\ &= 1 + n - n / k,
  \end{align*}
  which completes the induction.

  We have proved that 
  $\Exp_r \left[ \widetilde E_r(x; \phi) \right] \le n - n/k + 1$
  for any $x \in T$,
  where $r$ denotes the internal randomness of $\widetilde E$.
  In particular, 
  the same holds even if the expectation is
  taken over $x \in T$ as well as internal randomness of $E$.
  By averaging, there exists some internal randomness $r'$ of $E$
  such that $\Exp_{x \usample T} \left[ \widetilde E_{{r'}}(x; \phi) \right] \le n - n/k + 1$.
  Thus, $\widetilde E_{{r'}}(\blank; \phi) \colon T \to \binset^*$ is
  a deterministic prefix-free encoding whose
  average code length is at most $n - n / k + 1$.
  By Kraft's inequality and Jensen's inequality,
  we have
  \begin{align*}
    1 & \ge \sum_{x \usample T} 2^{-|\widetilde E_{{r'}}(x; \phi)|}
    = |T| \cdot \Exp_{x \usample T} \left[ 2^{-|\widetilde E_{{r'}}(x; \phi)|} \right]
    \ge |T| \cdot 2^{ \Exp_{x \usample T} \left[ -|\widetilde E_{{r'}}(x; \phi)| \right] }
    \ge |T| \cdot 2^{- (n - n/k + 1)}
  \end{align*}
  Hence, $|T| \le 2^{n - n / k + 1}$.
\end{proof}

\begin{corollary}
  \label{cor:parityapprox}
  Let $n \in \Nat$ and $0 \le \epsilon < \frac{1}{2}$.
  Let $s$ be the minimum size of a CNF formula $\phi$ that computes $\PAR_n$
  all but an $\epsilon$ fraction of inputs
  with one-sided error in the sense that $\phi^{-1}(1) \subset \PAR_n^{-1}(1)$.
  Then, for $k := \log \frac{2}{1 - 2\epsilon}$ and $k' = \floor{k}$,
  \[
    2^{\frac{n}{k + 1} - 3} \le s \le k' 2^{\ceil{\frac{n}{k'}} - 1}.
  \]
\end{corollary}
\begin{proof} %[Proof of Corollary \ref{cor:parityapprox}]
  We first prove the lower bound.
  Since $\phi^{-1}(1) \subset \PAR_n^{-1}(1)$, any satisfying assignment of $\phi$ is an isolated solution.
  Thus, by Lemma \ref{lemma:sat}, 
  $\Pr_{x \usample \PAR_n^{-1}(1)} [ \phi(x) = 1 ] \le 2^{n - n / (\ceil{\log s} + 2) + 1} / 2^{n - 1}
  \le 2^{- n / (\log s + 3) + 2}$.
  By the assumption,
  \[\epsilon \ge \Pr_{x \usample \binset^n} [ \phi(x) \neq \PAR_n(x) ] = \frac{1}{2} \Pr_{x \usample \PAR_n^{-1}(1)} [ \phi(x) = 0 ] \ge \frac{1}{2}(1 - 2^{-n/(\log s + 3) + 2}),\]
  and thus $s \ge 2^{\frac{n}{k + 1} - 3}$.
  
  For the upper bound, we divide an $n$-bit input into
  $k'$ blocks the $i$th block of which contains $n_i$ bits
  so that $\sum_{i=1}^{k'} n_i = n$ and $n_i \le \ceil{n / k'}$.
  For the $i$th block, we add $2^{n_i-1}$ clauses that check
  whether the parity of the block is $0$ if $i\neq1$ and $1$ if $i=1$.
  The constructed CNF formula $\phi$ is of size $\sum_{i=1}^{k'} 2^{n_i-1} \le k' 2^{\ceil{\frac{n}{k'}} - 1}$
  and does not accept any input whose parity is even.
  The number of the satisfying assignments of $\phi$ is exactly equal to $2^{n-k'}$,
  and hence $\Pr_x [\phi(x) \neq \PAR_n(x)] = \frac{1}{2} (1 - 2^{-k' + 1}) \le \epsilon$.
\end{proof}

Now we can determine the depth-3
formula size for computing $\PAR_n$.
\begin{theorem}
  \label{thm:paritymain}
  $\formulathree(\PAR_n) = \widetilde \Theta(2^{2 \sqrt n})$.
  More precisely, 
  $2^{ 2 \sqrt n - 5} \le \formulathree(\PAR_n) \le O( 2^{ 2\sqrt n} \sqrt n)$.
\end{theorem}
\begin{proof} %[Proof of Theorem \ref{thm:paritymain}]
  Let $\phi$ denote any CNF formula such that $\phi^{-1}(1) \subset \PAR_n^{-1}(1)$.
  Let $T$ denote the set of satisfying assignments of $\phi$.
  Since $T$ corresponds to the set of isolated solutions of $\phi$,
  by Theorem \ref{thm:isolated}, it follows that $|T| \le 2^{n - n / (\ceil{\log |\phi|}+2) + 1} \le
  2^{n - n / (\log |\phi| + 3) + 1}$.
  Therefore,
  for the uniform distribution $\mu$ on $\PAR_n^{-1}(1)$,
  the one-sided correlation per size is 
  \begin{align*}
    \Pr_{x \sample \mu} \left[ \phi(x) = 1 \right] / |\phi|
    & = |T|\, /\, ( 2^{n-1} \cdot |\phi| )
    \\ &
    \le
    2^{5 - n / (\log |\phi| + 3) - (\log |\phi| + 3)}
    %\\ &
    \le
    2^{5 - 2 \sqrt{n}},
  \end{align*}
  where the last inequality follows from the arithmetic and geometric means.
  That is, $\cor(\PAR_n) \le \cor[\mu](\PAR_n) \le 2^{5 - 2 \sqrt{n}}$
  and thus $2^{ 2 \sqrt n - 5} \le (\cor(\PAR_n))^{-1} \le \formulathree(\PAR_n)$
  by Theorem \ref{thm:correlation}.

  A depth-3 formula for computing $\PAR_n$
  can be constructed as follows:
  Divide $n$ variables into $n / k$ blocks each of which contains $k$ variables.
  For each string $z \in \binset^{n / k}$
  such that $\PAR_{n/k} (z) = 1$,
  we construct a CNF formula $\phi_z$ that checks
  whether, for each $i \in \NumSet{n / k}$, the parity of the $i$th block is $z_i$.
  Then it is easy to see that $\PAR_{n} = \bigvee_{z} \phi_z$.
  Since $|\phi_z| \le n / k  \cdot 2^{k - 1}$,
  we have $\formulathree(\PAR_n) \le n / k \cdot 2^{k + n / k - 1}$.
  Choosing $k := \sqrt{n}$, we obtain $\formulathree(\PAR_n) \le 2^{2\sqrt n - 1} \sqrt n$.
\end{proof}

\section{Concluding Remarks} \label{sec:conclusions}
In this paper, we
proved that depth-3 circuit size and depth-3 formula size are quadratically different.
Rossman \cite{Ros15} separated the depth-$d$ circuits from the depth-$d$ formulas for
a sufficiently large $d$ (say, $d \ge 108$).
These results leave a mysterious regime $3 < d \ll 108$ as an interesting open problem.
Can one prove that depth-$d$ circuits are more powerful than depth-$d$ formulas
for such a regime, and in particular, for $d = 4$?

Can one prove a depth-3 formula lower bound of $\PAR_n$ that is tight up to
a constant factor?
Can one extend our lower bound of approximating $\PAR_n$ by CNF formulas with one-sided error
to the case of two-sided error?
What is the asymptotic behavior of $T(n, n/2)$?
The techniques from the literature of Tur\'an's theorem might be helpful here.

Similarly, the techniques of O'Donnell and Wimmer \cite{OW07}
who studied approximation of the majority function by DNF formulas
(under the uniform distribution)
might be helpful for studying $T(n, n/2)$ and the depth-3 circuit size of majority.
We believe that the explicit link established between depth-3 formula size and
the complexity of approximation by depth-2
is useful to advance both of these lines of research further.
Can one find another application of the duality?

\bibliographystyle{alpha}
\bibliography{bib}

\end{document}